\documentclass{iopart}
\usepackage[T1]{fontenc}
\usepackage[utf8]{inputenc}
\usepackage[english]{babel}
\usepackage{graphicx}
\usepackage{setspace}
\usepackage{amssymb}
\usepackage{iopams}
\usepackage{setstack}
\usepackage{amsthm}
\usepackage[all]{xy}
\usepackage{subfig}

\theoremstyle{plain}
\newtheorem{thm}{Theorem}[section]
  \theoremstyle{plain}
  \newtheorem{cor}[thm]{Corollary}
  \theoremstyle{plain}
  \newtheorem{prop}[thm]{Proposition}
  \theoremstyle{remark}
  \newtheorem{rem}[thm]{Remark}
\newcommand{\OO}{\mathcal{O}}
\newcommand{\LL}{\mathcal{L}}
\newcommand{\PP}{\mathbb{P}}
\newcommand{\CC}{\mathbb{C}}
\newcommand{\tst}{T^{*}S}
\newcommand{\C}{\mathcal{C}}

\newcommand{\ZZ}{\mathbb{Z}}
\newcommand{\RR}{\mathbb{R}}
\newcommand{\lgl}{\tilde{\mathfrak{gl}}}
\newcommand{\M}{\mathcal{M}}
\newcommand{\norm}[1]{\left\Vert #1\right\Vert }

\newcommand{\TT}{\mathbb{T}}
\newcommand{\PG}{\PP G}

\begin{document}

\title{Algebraic integrability of confluent Neumann system}

\author{Martin Vuk }

\address{University of Ljubljana, Faculty of Computer and Information Science,
Tržaška cesta 25, SI-1001 Ljubljana, Slovenia}

\ead{\mailto{martin.vuk@fri.uni-lj.si}}

\begin{abstract}
In this paper we study the Neumann system, which describes the harmonic
oscillator (of arbitrary dimension) constrained to the sphere. In
particular we will consider the confluent case where two eigenvalues
of the potential coincide, which implies that the system has $S^{1}$
symmetry. We will prove complete algebraic integrability of confluent
Neumann system and show that its flow can be linearized on the generalized
Jacobian torus of some singular algebraic curve. The symplectic reduction
of $S^{1}$ action will be described and we will show that the general
Rosochatius system is a symplectic quotient of the confluent Neumann
system, where all the eigenvalues of the potential are double. This
will give a new mechanical interpretation of the Rosochatius system. 
\end{abstract}

\ams{37J35, 37J15, 14H70, 14H40}
\submitto{\JPA}
\maketitle

\section{Introduction}

Neumann system describes motion of a particle constrained to the $n$-dimensional
sphere $S^{n}$ under quadratic potential. The potential is given
in ambient coordinates $q=(q_{1},\ldots,q_{n+1})\in\RR^{n+1}$ by
the potential matrix $A=\mathrm{diag}(a_{1},\ldots,a_{n+1})$ as

\[
V(q)=\frac{1}{2}\langle Aq,q\rangle=\frac{1}{2}\sum_{i=1}^{n+1}a_{i}q_{i}^{2}.
\]
In the generic case where all the eigenvalues of the potential $a_{i}$
are different, the Neumann system is algebraically completely integrable
and its flow can be linearized on the Jacobian torus of an algebraic
spectral curve \cite{Moser:CSY:1980,Adler:AdvM:1980,Mumford:TataII:1984}.
A standard approach to study integrable systems is by writing down
the system in the form of Lax equation, which describes the flow of
matrices or matrix polynomials with constant eigenvalues, i.e. the
isospectral flow \cite{Adler:AdvM:1980,audin:CUP:1999}. Eigenvalues
of the isospectral flow are the first integrals of the integrable
system and Lax representation maps Arnold-Liouville tori into the
real part of the isospectral manifold consisting of matrices with
the same spectrum. A quotient of the isospectral manifold by a suitable
gauge group is in turn isomorphic to the open subset of the Jacobian
of the spectral curve \cite{Beauville:Acta:1990}. 

Two different Lax equations are known for a generic Neumann system
with $n$ degrees of freedom: one is using $(n+1)\times(n+1)$ matrix
polynomials of degree $2$ \cite{Adler:AdvM:1980} and the other is
using $2\times2$ matrix polynomials of degree $n$ \cite{Mumford:TataII:1984,Beauville:Acta:1990}.
The $(n+1)\times(n+1)$ Lax equation was used by Audin \cite{Audin:ExM:1994}
to describe the Arnold-Liouville tori for geodesic motion on an ellipsoid,
which is equivalent to the Neumann system.

In contrast to the generic Neumann system, the special confluent case
in which some eigenvalues of the potential coincide has not received
much attention. In this paper we will consider the confluent case
with two of the eigenvalues being the same and the Neumann system
having an additional $S^{1}$ rotational symmetry. We will show that
the confluent Neumann system is algebraically completely integrable
and that its flow can be linearized on the generalized Jacobian of
a singular algebraic curve. We will describe the symplectic reduction
of the $S^{1}$ action, which will yield an alternative description
of the Rosochatius system as a symplectic quotient of the confluent
Neumann system. Mechanically speaking, Rosochatius system can be seen
as a Neumann system on a rotating sphere. Inspired by this result
we will describe general Rosochatius system as a reduction of confluent
Neumann system with all eigenvalues of the potential being double.
This result combined with the proof of integrability of the confluent
Neumann system will also give an alternative proof of the algebraic
integrability of the Rosochatius system.

When applying the $2\times2$ Lax equation to the confluent case,
the resulting Lax equation in fact describes the Rosochatius system
and not the confluent Neumann system \cite{Adams:Harnad:CMP:1988,Hurtubise:CMP:1990,Harnad:LNPhys:1993}.
We will therefore use $(n+1)\times(n+1)$ Lax equation, where the
resulting spectral curve is singular in the confluent case. Following
the standard procedure and normalizing the spectral curve results
in the loss of one degree of freedom. In order to avoid that, we will
use the generalized Jacobian of the singular spectral curve to linearize
the flow as in \cite{Gavrilov:CRMProc:2002,Vivolo:EdMathSoc:2000}.
The generalized Jacobian is an extension of the {}``ordinary'' Jacobian
by a commutative algebraic group (see \cite{Serre:Springer:ClassFields}
for more detailed description). In our case the extension will be
the group $\CC^{*}$ that corresponds to the rotational symmetry of
the initial system. Generalized Jacobian was used by others to linearize
the flow of other integrable systems with rotational symmetry for
example spherical pendulum or Lagrange top \cite{Zhivkov:Ens:1998,Vivolo:EdMathSoc:2000,audin:CMP:2002,Gavrilov:CRMProc:2002}.
Our study however will give a more detailed description of the relation
between symplectic reduction and algebraic reduction from generalized
to ordinary Jacobian.

After the introduction, Hamiltonian reduction and Liouville integrability
of the confluent Neumann system are discussed in section \ref{sec:Symmetric-Neumann-system}.
In section \ref{sec:Lax-representation-of} we study $(n+1)\times(n+1)$
Lax equation and corresponding isospectral manifolds. Our main result
is formulated in theorem \ref{thm:Arnold-Jacobian} and describes
the relation of Arnold-Liouville tori to the generalized Jacobian
of the singular spectral curve. As a corollary the complete algebraic
integrability and Liouville integrability of the confluent Neumann
system will follow. We conclude with proposition \ref{pro:bif-diag},
which describes the bifurcation diagram of the energy momentum map
in terms of algebraic data.

\section{Hamiltonian description\label{sec:Symmetric-Neumann-system}}

The Neumann system describes a particle on a sphere (of arbitrary
dimension) under the influence of quadratic potential. We can write
it as a Hamiltonian system on the cotangent bundle of the sphere $\tst^{n}$
with canonical symplectic form $\omega_{c}$ and the Hamiltonian $H$
given in ambient coordinates $(q,p)\in\RR^{n+1}\times\RR^{n+1}$ as

\[
H(q,p)=\frac{1}{2}\left(\norm{p}^{2}+\langle q,Aq\rangle\right).
\]
The potential is given by a positive definite linear operator $A$
on $\RR^{n+1}$. For simplicity we will assume that $A$ is diagonal
with positive eigenvalues $a_{i}$. A consequence of positivity is
that the Hamiltonian is proper and the energy level sets - and hence
Arnold-Lioville tori - are compact. The equations of motion in Hamiltonian
form are 
\begin{eqnarray}
\dot{q} & = & p\label{eq:hamiltonian-system}\\
\dot{p} & = & -Aq+\varepsilon q\nonumber 
\end{eqnarray}
 where $\varepsilon=\norm{p}^{2}+c$ is chosen so that $\norm{q}=1$
and the particle stays on the sphere.

\subsection{Reduction of the symmetry}

Let us consider the confluent case where all the eigenvalues $a_{i}$
of the potential $A$ are distinct, except $a_{n}=a_{n+1}$. The action
$\varphi'$ of $S^{1}=SO(2,\RR)$ on $S^{n}$, given by rotations
in the $q_{n},q_{n+1}$ plane, leave the potential invariant and can
be lifted to the symplectic action 

\[
\varphi:\tst^{n}\times S^{1}\to\tst^{n}
\]
on the cotangent bundle that leaves the Hamiltonian $H$ invariant.
We would like to reduce this action and describe the resulting reduced
system in more detail.

Let $K$ be the moment map for the lifted action $\varphi$. The map
$K$ is a real function on $\tst^{n}$ as $\mathfrak{so}(2)^{*}\simeq\RR$
and is the \emph{angular momentum} for the rotations in $q_{n},q_{n+1}$
plane 
\[
K(q,p)=q_{n}p_{n+1}-q_{n+1}p_{n}.
\]
Since the action $\varphi$ is not free, the description of reduced
system is more complicated. Let us denote with $F$ the set of points
on the sphere, where the action $\varphi'$ is not free. This is precisely
the fixed point set given by the condition 
\[
r^{2}=q_{n}^{2}+q_{n+1}^{2}=0
\]
The set $F$ is a codimension-$2$ great sphere on $S^{n}$. Note
that on $F$, the value $k$ of the moment map $K$ equals $0$. The
set of regular points $S_{reg}^{n}=S^{n}-F$ is an open subset of
$S^{n}$ on which the action is free. 

The reduced system is defined on symplectic quotients $(M_{k},\omega_{k})$,
which are parametrized by the value $k$ of the moment map $K$. We
will use the operator $//_{k}$ to denote the symplectic quotient.
By definition 
\[
M_{k}=\tst^{n}//_{k}S^{1}:=K^{-1}(k)/S^{1}
\]
 and $\omega_{k}$ is defined with $\pi^{*}\omega_{k}=\omega_{c}|_{K^{-1}(k)}$,
where $\pi$ is the quotient projection. For $k\not=0$ the fiber
$K^{-1}(k)$ is a corank-$1$ sub-bundle of the cotangent bundle over
the set of regular points $\tst_{reg}^{n}$. Since the action is free
on $K^{-1}(k)$, the standard result for lifted actions gives 
\[
(M_{k},\omega_{k})=(T^{*}(S_{reg}^{n}/S^{1}),\omega_{c}+\omega_{p}),
\]
where $\omega_{p}$ is the magnetic term coming from the curvature
of the mechanical connection. In our case the mechanical connection
is flat and $\omega_{p}=0$ \cite{Marsden:LMS:1992}. The quotient
manifold $S_{reg}^{n}/S^{1}$ is a dimension-$(n-1)$ open half sphere
$S_{+}^{n-1}$ as we can see if we introduce cylindrical coordinates
$(q_{1},\ldots,q_{n-1},r(\cos\varphi,\sin\varphi))$. The set $S_{reg}^{n}$
is given by the condition $r\not=0$ and the quotient $S_{reg}^{n}/S^{1}$
can be parametrized by $(q_{1},\ldots,q_{n-1},r)\in S^{n-1}$ for
$r>0$. We can also see directly by using cylindrical coordinates
that the perturbation $\omega_{p}$ of the canonical symplectic form
is zero.

For $k=0$ the description of $M_{0}$ is more complicated, since
the action on $K^{-1}(0)$ is not free, and $M_{0}$ is not a manifold.
We will remedy this by considering a singular double cover of the
reduced phase space instead.

The set $K^{-1}(0)$ is not a sub-bundle of $\tst^{n}$, because at
fixed points $r=0$ the fiber has full rank and over $S_{reg}^{n}$
the fiber has codimension $1$. However, we can still pass the quotient
on the base manifold because the action is a cotangent lift. The quotient
space $S^{n}/S^{1}$ is a dimension-$(n-1)$ closed half sphere $(S_{+}^{n-1})^{c}$
and can be parametrized by cylindrical coordinates $(q_{1},\ldots,q_{n-1},r)\in S^{n-1}$,
with $r\ge0$. The symplectic quotient $M_{0}$ is a cotangent bundle
over $S^{n}/S^{1}$, in the sense that the fiber at a singular point
with $r=0$ is a closed half plane of {}``positive differentials''
(the differential, which is positive at the directions normal to the
boundary of $S^{n}/S^{1}$). In that sense we can write
\[
(M_{0},\omega_{0})=(T^{*}(S_{+}^{n-1})^{c},\omega_{c}).
\]

To conclude the construction of the reduced system we have to calculate
the reduced Hamiltonian. This can be easily done using the parametrization
with cylindrical coordinates $(\hat{q},\hat{p})=((q_{1},\ldots,q_{n-1},r),(p_{1},\ldots,p_{n-1},p_{r}))$
\begin{eqnarray*}
H(q,p) & = & \frac{1}{2}\left(\sum_{i=1}^{n-1}\left(p_{i}^{2}+a_{i}q_{i}^{2}\right)+p_{n}^{2}+p_{n+1}^{2}+a_{n}(q_{n}^{2}+q_{n+1}^{2})\right)=\\
 & = & \frac{1}{2}\left(\sum_{i=1}^{n-1}p_{i}^{2}+p_{r}^{2}+\frac{k^{2}}{r^{2}}+\sum_{i=1}^{n-1}a_{i}q_{i}^{2}+a_{n}r^{2}\right)=\\
 & = & \frac{1}{2}\left(\norm{\hat{p}}^{2}+V_{k}(\hat{q})+\langle\hat{q},\hat{A}\hat{q}\rangle\right)=H_{r}(\hat{q},\hat{p}),
\end{eqnarray*}
where $\hat{A}=\text{diag}(a_{1},\ldots,a_{n})$. The additional term
$V_{k}(\hat{q})=\frac{k^{2}}{r^{2}}$ corresponds to the centrifugal
system force, induced by the rotation. 

To simplify the description of the reduced system, we will use the
whole sphere $S^{n-1}$ instead of only half of it. The half sphere
$(S_{+}^{n-1})^{c}$ can be also viewed as a quotient of $S^{n-1}$
by the group $\ZZ_{2}$ acting with reflections $r\mapsto-r$. The
reduced Hamiltonian $H_{r}(\hat{q},\hat{p})$ can be lifted on $(\tst^{n-1},\omega_{c})$,
since it only depends on $r^{2}$ and $p_{r}^{2}$ and is invariant
for $\ZZ_{2}$ action. We have thus constructed a Hamiltonian system
$(\tst^{n-1},\omega_{c},H_{r})$ such that its reduction coincides
with the reduced Neumann system. 

\begin{thm}
The singular symplectic quotient $(M_{k},\omega_{k},H_{r})$ of the
Neumann Hamiltonian system $(\tst^{n-1},\omega_{c},H_{A})$ with potential
given by the matrix $A=\text{diag}(a_{1},\ldots,a_{n},a_{n})$ is
isomorphic to the quotient by the $\ZZ_{2}$ action of the perturbed
Neumann system $(\tst^{n-1},\omega_{c},H_{\hat{A}}+V_{k})$ with potential
matrix $\hat{A}=\mathrm{diag}(a_{1},\ldots,a_{n})$.
\end{thm}
The symplectic quotient of confluent Neumann system is in fact a special
case of Rosochatius system, which was studied in \cite{Adams:Harnad:CMP:1988}
for example. A general Rosochatius system is a Hamiltonian system
on $\tst^{n}$ with potential
\begin{equation}
V=\frac{1}{2}\left(\sum_{i=1}^{n+1}a_{i}q_{i}^{2}+\frac{k_{i}^{2}}{q_{i}^{2}}\right).\label{eq:rosochatius}
\end{equation}
We can generalize the procedure from the above to the case where potential
matrix has all the eigenvalues double $A=\mathrm{diag}(a_{1},a_{1},a_{2},a_{2},\ldots,a_{n+1},a_{n+1})$.
This gives a mechanical interpretation of Rosochatius system as Neumann
system on a rotating sphere. 

\begin{cor}
The symplectic quotient of the confluent Neumann system with potential
matrix $A=\mathrm{diag}(a_{1},a_{1},a_{2},a_{2},\ldots,a_{n+1},a_{n+1})$
by the group of rotations $(S^{1})^{n+1}$ is isomorphic to the quotient
by the group of reflections $\ZZ_{2}^{n+1}$ of the Rosochatius system
on $\tst^{n}$ given by the potential \eref{eq:rosochatius}. The
coefficients $k_{i}$ of the rational part are the values of the angular
momentum of rotations in corresponding eigenplanes of the potential
matrix $A$.
\end{cor}

\subsection{Integrability}

Hamiltonian system with Hamiltonian $H$ on the symplectic manifold
of dimension $2n$ is called \emph{integrable }if there exist $n$
functionally independent pairwise Poisson commutative first integrals,
one of which is $H$. To show that the system is \emph{completely
algebraically integrable, }we will prove that the level sets of the
first integrals are real parts of the extensions of Abelian varieties
by $\CC^{*}$ (see \cite{Mumford:TataII:1984} for a definition of
complete algebraic integrability). 

The integrability of the confluent Neumann system is a consequence
of the integrability of the generic Neumann system, which is well
known\cite{Moser:CSY:1980,Adler:VanMoerbeke:AdvM:1980}. The first
integrals for the confluent case can be obtained by taking the limit
$a_{n}\to a_{n+1}$ on Uhlenbeck's integrals for generic Neumann system:
\begin{equation}
F_{i}^{g}=q_{i}^{2}+\sum_{j\not=i}\frac{(q_{i}p_{j}-q_{j}p_{i})^{2}}{a_{j}-a_{i}}.\label{eq:uhlebeck}
\end{equation}
A set of commuting integrals, equivalent to \eref{eq:uhlebeck},
tends to a set of commuting integrals for the confluent case 
\begin{eqnarray*}
\{F_{1}^{g},\ldots,F_{n-1}^{g},F_{n}^{g}+F_{n+1}^{g},\frac{1}{2}(a_{n}-a_{n+1})(F_{n}^{g}-F_{n+1}^{g})\}\to\\
\to\{F_{1},\ldots,F_{n-1},F_{n},K^{2}\}
\end{eqnarray*}
when taking the limit $a_{n}\to a_{n+1}$. The set of commuting integrals
for the confluent Neumann system is given by
\begin{eqnarray}
F_{i} & = & q_{i}^{2}+\sum_{j\not=i}^{}\frac{(q_{i}p_{j}-q_{j}p_{i})^{2}}{a_{j}-a_{i}};\quad i<n\nonumber \\
F_{n} & = & q_{n}^{2}+q_{n+1}^{2}+\sum_{j<n}\frac{(q_{j}p_{n}-q_{n}p_{j})^{2}+(q_{j}p_{n+1}-q_{n+1}p_{j})^{2}}{a_{j}-a_{n}}\label{eq:ul_int}\\
K^{2} & = & (q_{n}p_{n+1}-q_{n+1}p_{n})^{2}.\nonumber 
\end{eqnarray}
Note that $K$ is angular momentum for the rotations in $q_{n},q_{n+1}$
plane. We can also verify that the integrals in \eref{eq:ul_int}
are not independent but satisfy the same relations as Uhlenbeck's
integrals in the generic case
\[
\sum_{i=1}^{n}F_{i}=1
\]
and that the Hamiltonian $H$ can be expressed as a linear combination
of $F_{i}$ and $K^{2}$ 
\[
\sum_{i=1}^{n}a_{i}F_{i}+K^{2}=2H.
\]
The Poisson brackets of $F_{i}$ are continuous functions of $a_{i}$
so the commutativity of the integrals is preserved when taking the
limit $a_{n}\to a_{n+1}$. Commutativity of $K$ with $F_{i}$ also
follows from the fact that $F_{i}$ are invariant for $S^{1}$ action,
generated by $K$. To conclude the proof of integrability one needs
to verify the independence of the integrals \eref{eq:ul_int} (up
to relation $\sum F_{i}=1$). The commutativity of the integrals $F_{i}$
and $K^{2}$ also follows from the AKS theorem, if we write $F_{i}$
an $K^{2}$ as invariant functions of Lax matrix on appropriate loop
algebra. This is a standard way to prove commutativity of Uhlenbeck's
integrals for the generic Neumann case \cite{Adams:Harnad:CMP:1988,Beauville:Acta:1990}.

Let us combine all the first integrals in a map 
\begin{eqnarray*}
F_{EM}: & \tst^{n}\to\RR^{n}\\
 & (q,p)\mapsto(F_{1},\ldots,F_{n-1},K)
\end{eqnarray*}
we will call \emph{energy momentum map}. The fundamental property
of integrable systems is that the level sets of energy momentum map
$F_{EM}$ are n-dimensional tori, on which the flow can be linearized.
We will also use the complexified version of $F_{EM}$, defined on
$(\tst^{n})^{\CC}$. 

\begin{thm}
The confluent Neumann system $(\tst^{n},\omega_{c},H)$ is algebraically
completely integrable system.
\end{thm}
We have seen that by taking the limit $a_{n}\to a_{n+1}$, $\{F_{1},\ldots,F_{n-1},K^{2}\}$
is a set of $n$ commuting first integrals for symmetric Neumann system
and we will show later in \ref{ssec:Linearisation-of-the} that they
are functionally independent. We will also prove the part about algebraically
complete integrability in subsection \ref{ssec:Linearisation-of-the}.

From the fact that the symmetric Neumann system is integrable also
follows that its symplectic quotient is integrable. This gives alternative
proof of the integrability of Rosochatius system, which is a symplectic
quotient of the symmetric Neumann system. 

\begin{cor}
The Rosochatius system $(\tst^{n},\omega_{c},H+V_{r})$ with $V_{r}=\sum\frac{ki^{2}}{q_{i}^{2}}$
is algebraically completely integrable system.
\end{cor}

\section{Lax representation of the confluent Neumann system\label{sec:Lax-representation-of}}

In this section we will use Lax equation for $(n+1)\times(n+1)$ matrix
polynomials to study confluent Neumann system with $n$ degrees of
freedom. We will show that the flow of the system can be linearized
on the Jacobian of the singular spectral curve and that the system
is completely algebraically integrable. This will also conclude the
proof of Liouville integrability and yield the description of bifurcation
diagram for energy momentum map.

\subsection{Lax equation in $\lgl(n+1,\CC)$}

Let us write the confluent Neumann system \eref{eq:hamiltonian-system}
as an isospectral flow of matrix polynomials in loop algebra $\lgl(n+1,\CC)$.
We will use the Lax equation introduced by Moser \cite{Moser:CSY:1980}.
The \emph{loop algebra} $\lgl(n+1,\CC)$ consists of Laurent polynomials
with coefficients in $\mathfrak{gl}(n+1,\CC)$, and can be written
as a tensor product $\mathfrak{gl}(n+1,\CC)\otimes\CC[\lambda,\lambda^{-1}]$.
Consider complexified Neumann system \eref{eq:hamiltonian-system}
on a subspace $(\tst^{n})^{\CC}\subset\CC^{n+1}\times\CC^{n+1}$ defined
by constraints $\sum q_{i}^{2}=1$ and $\sum q_{i}p_{i}=0$. We introduce
\emph{Lax matrix polynomial} from $\lgl(n+1,\CC)$
\begin{equation}
L(\lambda)=A\lambda^{2}+q\wedge p\ \lambda-q\otimes q\ ;\quad q,p\in\CC^{n+1}\label{eq:lax_n+1}
\end{equation}
where $A$ is the potential matrix of the Neumann system and $(q,p)\in(\tst^{n})^{\CC}$.
Neumann system either generic or confluent can be written as \emph{Lax
equation} 
\begin{equation}
\frac{\rmd}{\rmd t}L(\lambda)=[M(\lambda),L(\lambda)].\label{eq:Lax}
\end{equation}
with $M(\lambda)=A\lambda+q\wedge p=(\lambda L(\lambda))_{+}$ (the
subscript $_{+}$ denotes the polinomial part of an element of $\lgl(n+1,\CC)$)\cite{Adler:AdvM:1980}.
The flow of \eref{eq:Lax} is isospectral as it conserves the spectrum
of the matrix $L(\lambda)$. This means that the characteristic polynomial
$P(\lambda,\mu)$ and the corresponding affine spectral curve, defined
by equation 
\[
P(\lambda,\mu)=\det(L(\lambda)-\mu)=0
\]
do not change along the flow and can be expressed only with the values
of the first integrals of Neumann system. In order to avoid unnecessary
singularities at the infinity, the affine curve is completed in the
total space of the line bundle $\OO_{\PP^{1}}(2)$ over $\PP^{1}$,
which is given by the transition function $(\lambda,\mu)\mapsto(\lambda^{-1},\lambda^{-2}\mu)$
. The completion of the affine spectral curve in $\OO_{\PP^{1}}(2)$
is called \emph{the spectral curve }of $L(\lambda)$ and denoted by
$\C_{m}$. Note that for any given $L(\lambda)$ there are also naturally
defined a map $\lambda:\C_{m}\to\PP^{1}$ and a section $\mu$ of
the bundle $\lambda^{*}\OO_{\PP^{1}}(2)$ apart from the spectral
curve $\C_{m}$.

The spectral curve of $L(\lambda)$ is hyperelliptic and we can see
that by introducing new variables $x=\lambda^{-2}\mu$ and $y=\lambda\prod_{i=1}^{n+1}(a_{i}-x)$,
where $a_{i}$ are the eigenvalues of the potential matrix $A.$ The
equation we obtain is:
\begin{equation}
y^{2}=Q(x)\prod_{i=1}^{n+1}(a_{i}-x),\label{eq:sing_hyper_eq}
\end{equation}
where $Q(x)$ is a polynomial of degree $n$ \cite{Moser:CSY:1980,Audin:ExM:1994}.
The coefficients of $Q(x)$ are first integrals of the confluent Neumann
system and if we write $Q(x)$ in terms of Lagrange interpolating
polynomials over the points $a_{1},\ldots,a_{n},a_{n}$ using notation
with partial fractions 
\begin{equation}
Q(x)=\prod_{j=1}^{n+1}(a_{j}-x)\left(\frac{F_{n}}{a_{n}-x}+\frac{K^{2}}{(a_{n}-x)^{2}}+\sum_{i=1}^{n-1}\frac{F_{i}}{a_{i}-x}\right),\label{eq:qx}
\end{equation}
the coefficients we obtain are the integrals \eref{eq:ul_int} we
met before . 

In the confluent case, where $a_{n}=a_{n+1}$, the product $\prod(a_{i}-x)$
has a quadratic factor $(a_{n}-x)^{2}$ and the spectral curve has
a singular point $S$ given by $(x,y)=(a_{n},0)$ or $(\lambda,\mu)=(\infty,a_{n})$.
The singularity $S$ is a double point for $K\not=0$ and a cusp for
$K=0$. The smoothness of the spectral curve is closely related to
the regularity of the matrix $L(\lambda)$. Recall that a matrix $B\in\mathfrak{gl}(r,\CC)$
is called \emph{regular} if all the eigenspaces of $B$ are one dimensional.

\begin{prop}
\label{pro:smooth->regular}Let $\C$ be the spectral curve of matrix
polynomial $L(\lambda)$ and $a\in\PP^{1}$. If all the points $\lambda^{-1}(a)\in\C$
are smooth, then the matrix $L(a)$ is regular.
\end{prop}
For proof see \cite{Beauville:Acta:1990,Gavrilov:CRMProc:2002}. The
value of $L(\lambda)\lambda^{-2}$ at $\lambda=\infty$ is the matrix
$A$ and the singularity at $S$ is a consequence of the fact that
$A$ is not regular when $a_{n}=a_{n+1}$ . 

Let $\C$ be normalization of the spectral curve $\C_{m}$, which
is described by\begin{equation}
w^{2}=Q(x)\prod_{i=1}^{n-1}(a_{i}-x),\label{eq:norm_hyper_eq}\end{equation}
where $w=y/(a_{n}-x)$. We will call $\C$ the \emph{normalized spectral
curve }of $L(\lambda)$. There is a map $\pi:\C\to\C_{m}$ that is
biholomorphic everywhere except at the inverse image of the singular
point $S$. The inverse image $\pi^{-1}(S)$ consists of two points
$\{P_{+},P_{-}\}$ for $K\not=0$ and a point $P_{0}$ for $K=0$.
In case when $K\not=0$, the curve $\C_{m}$ is obtained from $\C$
by identifying the points $\{P_{+},P_{-}\}$ into the singular point
$S$. The singular curve $\C_{m}$ can be described as a singularization
of $\C$ given by \emph{modulus} $m$ (see \cite{Serre:Springer:ClassFields}
for details). The modulus is $m=P_{+}+P_{-}$ for $K\not=0$ and $m=2P_{0}$
for $K=0$.

\begin{rem}
Note that we will only resolve the singularity at $S,$ which is {}``generic''
in the sense that it appears for all the values of the energy-momentum
map. So the curve $\C$ can still be singular for some values of the
energy-momentum map.$ $
\end{rem}
Finally we find the genus of $\C$ and $\C_{m}$ from the fact that
the curves are hyperelliptic and from the degree of the polynomials
in \eref{eq:sing_hyper_eq} and \eref{eq:norm_hyper_eq}. We obtain
$g(\C)=n-1$ for normalized spectral curve $\C$ and the arithmetic
genus $g_{a}(\C_{m})=n$ for the singular curve $\C_{m}$.

\subsection{Isospectral manifold of matrix polynomials}

We have seen that the Neumann system satisfies Lax equation and that
the spectral curve depends only on the first integrals and that all
the first integrals are encoded in the spectral curve. The level sets
of (complexified) energy momentum map $F_{EM}$ lie in the set of
matrices $L(\lambda)$ with fixed spectral curve $\C_{m}$. It is
therefore essential to describe the set of matrix polynomials with
a given spectral curve. 

Let $P(\lambda,\mu)$ be \emph{a spectral polynomial} for some matrix
polynomial $L(\lambda)\in\lgl(r,\CC)$). Denote with $\M_{P}$ the
subset of all the elements of $\lgl(r,\CC)$ with the same characteristic
polynomial $P$ (this also fixes the spectral curve $\C$)
\[
\M_{P}=\{L(\lambda);\quad\det(L(\lambda)-\mu)=P(\lambda,\mu)\}.
\]
All $L(\lambda)\in\M_{P}$ have the same spectral curve $\C$, which
is defined by $P(\lambda,\mu)=0$. 

While the characteristic polynomial and thus the eigenvalues of $L(\lambda)$
are fixed by the flow, the eigenvectors and eigenspaces change. Let
us define a map 
\[
\xi_{L(\lambda)}:\C\to\mathbb{P}^{r-1},
\]
such that $\xi_{L(\lambda)}((\lambda,\mu))$ is one dimensional eigenspace
of the matrix $L(\lambda)$ with respect to the eigenvalue $\mu$.
If the spectral curve is smooth, than by proposition \ref{pro:smooth->regular}
all the eigenspaces of $L(\lambda)$ for any $\lambda$ are onedimensional
and the map $\xi$ is well defined. The map $\xi_{L(\lambda)}$ defines
a line bundle on $\C$ and its dual is called \emph{eigenvector line
bundle }or shorter \emph{eigenline bundle. }We will denote eigenline
bundle by $\LL_{L(\lambda)}$. By construction, the eigenline bundle
is a subbundle of the trivial bundle $\C\times\CC^{r}$. One can see
by using Riemann-Roch-Grothendick theorem that the degree (Chern class)
$d$ of the eigenline bundle $\LL_{L(x)}$ equals $g+r-1$ where $g$
is the genus of the spectral curve $\C$.

The only condition for $\xi_{L(\lambda)}$ to be defined is that the
eigenspaces of $L(\lambda)$ are one dimensional for all but finite
number of points on $\C$. If the spectral curve is singular, than
the map $\xi_{L(\lambda)}$ can be defined on the set of points $\C_{m}-N\subset\C_{m}$
where the matrix $L(\lambda)$ has one dimensional eigenspace. If
the set $N$ is finite, the map $\xi_{L(\lambda)}$ can be extended
as a holomorphic map and eigenline bundle $\LL_{L(x)}$ can be defined
on the normalization $\C$ of the spectral curve. Note that by the
proposition \ref{pro:smooth->regular} the set $N$ is a subset of
singular locus of the spectral curve $\C_{m}$. 

One can define the \emph{eigenbundle map} \begin{eqnarray*}
e: & \M_{P}\to Pic^{d}(\C)\\
 & L(\lambda)\mapsto[\LL_{L(\lambda)}],\end{eqnarray*}
from $\M_{P}$ to the Picard group $Pic(\C)$ of isomorphism classes
of line bundles on the normalized spectral curve $\C$. The subset
$Pic^{d}(\C)$ consists of classes of line bundles with given Chern
class $d$. The set $Pic^{d}(\C)$ is a copy of the zero degree Picard
subgroup \texttt{$Pic^{0}(\C)$,} which is in turn isomorphic via
Abel-Jacobi map to the Jacobian $Jac(\C)$ of $\C$. The map $e$
assigns to each matrix polynomial $L(\lambda)$ the isomorphism class
of its\emph{ }eigenline bundle\emph{ }and thus encodes the flow of
Lax equation. The map $e$ is not surjective since eigenline bundles
cannot lie in the special divisor $\Theta$ on the Jacobian. The map
$e$ is neither injective since the class of eigenline bundle defines
the matrix polynomial only up to conjugation by the gauge group $\PP Gl(r,\CC)$.
The space we have to consider is the quotient space $\M_{P}/\PP Gl(r,\CC)$
and it was shown in \cite{Beauville:Acta:1990} that if the spectral
curve $\C$ is smooth, the space $\M_{P}/\PP Gl(r,\CC)$ is isomorphic
as an algebraic manifold to the Zariski open subset $Jac(\C)-\Theta$
of the Jacobian of the spectral curve $\C$. The isomorphism is given
by the eigenbundle map $e$.

As the leading coefficient in $L(\lambda)$ is preserved by the flow,
we will use the closed subset of $\M_{P}$ of matrix polinomials with
fixed leading coefficient. Let $L(\lambda)=A\lambda^{l}+A_{l-1}\lambda^{l-1}+\ldots+A_{0}$
and let $A$ be fixed. We denote 
\[
\M_{P}^{A}=\{L(\lambda)\in\M_{P};\quad\lim_{\lambda\to\infty}L(\lambda)/\lambda^{l}=A\}.
\]
The action of the gauge group $\PP Gl(r,\CC)$ on $\M_{P}$ reduces
to the action of the stabilizer subgroup $\PG_{A}<\PP Gl(r,\CC)$
of $A$. The quotient $\M_{P}^{A}$ by $\PG_{A}$ is again isomorphic
to the Jacobian
\[
\M_{P}^{A}/\PG_{A}\simeq Jac(\C)-\Theta.
\]

\begin{rem}
If the spectral curve $\C$ is smooth at infinity, the matrix $A$
has to be \emph{regular }by proposition \ref{pro:smooth->regular}\emph{.}
As a consequence, the stabilizer group of $A$ is the product $\PG_{A}=(\CC^{*})^{s-1}\times\CC^{r-s}$
where $s$ is the number of distinct eigenvalues of $A$. If $A$
is not regular, the dimension of the stabilizer group $\PG_{A}$ is
larger. In our case, when $A=\mathrm{diag}(a_{1},\ldots,a_{n},a_{n})$
the group $\PG_{A}$ is the product 
\[
(\CC^{*})^{n-2}\times Gl(2,\CC).
\]
\end{rem}
Let us assume for one moment that the matrix $A$ is regular. We have
seen before that the level set $\M_{P}^{A}$ is an extension of $Jac(\C)-\Theta$
by an Abelian algebraic group $\PG_{A}$. As it happens, the generalized
Jacobian is also defined as an extension by Abelian algebraic group
and it was shown in \cite{Gavrilov:CRMProc:2002} that 
\[
\M_{P}^{A}\simeq Jac(\C_{m})-\Theta
\]
for a suitable choice of modulus $m$. The chosen modulus is the effective
divisor consisting of infinite points $P_{i}$ on the spectral curve
\[
m=\sum_{P_{i}\in\lambda^{-1}(\infty)}m_{i}P_{i},
\]
where the coefficients $m_{i}$ are multiplicities of the eigenvalues
$\mu(P_{i})$ of $A=L(\infty)$. The isomorphism $\M_{P}^{A}\to Jac(\C_{m})-\Theta$
is given by the eigenbundle map $e_{m}$ to the generalized Jacobian.
Let $S=|m|$ be the set of points in $\C$ that are mapped to the
singular point in $\C_{m}$. The line bundle on the singular curve
$\C_{m}$ is given by a divisor on $\C$ that does not intersect the
singular set $S$. To give a line bundle on $\C_{m}$ is thus enough
to give a section of a line bundle on $\C$ that has no zeros or poles
in $S$. We can then extend the map $e:\M_{P}\to Pic^{d}(\C)$ to
a map 
\[
e_{m}:\M_{P}\to Pic^{d}(\C_{m})
\]
by choosing a section of $e(L(\lambda))$ uniformly on $\M_{P}$ that
does not have any zeros or poles in the infinite point. This can be
done by appropriate normalization. It was shown in \cite{Gavrilov:CRMProc:2002,Vivolo:EdMathSoc:2000}
that $e_{m}$ gives an isomorphism from the isospectral space $\M_{P}^{A}$
to the Zariski open subset $Jac(\C_{m})-\Theta$ of the generalized
Jacobian of the singular spectral curve, given by modulus $m$.

\begin{rem}
The above results give orientation about the expected number of degrees
of freedom of isospectral flows. We see that the upper limit is the
arithmetic genus\emph{ }of the singularization of the spectral curve,
which depends on the degree $l$ and the rank $r$ of $L(\lambda)$. 
\end{rem}
In the case of confluent Neumann system, the matrix $A$ is not regular
and corresponding spectral curve is singular. We will show later that
the Lax flow also preserves part of the {}``derivative'' of $L(\lambda)$
at the singular point. We will restrict the space $\M_{P}^{A}$ further
by fixing a specific block of the lower term $A_{l-1}$. The resulting
isospectral set will again be isomorphic to the open subset of the
generalized Jacobian \cite{Vivolo:EdMathSoc:2000}. Note that we have
to assume regularity of the fixed block of the lower term $A_{l-1}$
in order to have the isomorphism.

\subsection{Proof of the integrability \label{ssec:Linearisation-of-the}}

In order to prove the Liouville integrability of the confluent Neumann
system, we have to show that its first integrals $\{F_{1},\ldots,F_{n-1},K\}$
Poisson commute and that they are functionally independent. We already
know that the integrals commute from section \ref{sec:Symmetric-Neumann-system}.
We also know that the level sets of the first integrals - the Arnold-Liouville
tori - lie in the isospectral manifold $\M_{P}^{A}$. We will use
the eigenbundle map $e_{m}$ to map Arnold-Liouville tori to the real
part of the generalized Jacobian of the singular spectral curve $\C_{m}$
and show that this map is a $(\ZZ_{2})^{n-1}$ covering. This will
prove complete algebraic integrability of confluent Neumann system
and independence of the first integrals will follow. 

Let $E_{n}$ be the eigenspace of the double eigenvalue $a_{n}$,
which is spanned by the unit vectors $e_{n}$ and $e_{n+1}$. The
behavior of $L(\lambda)|_{E_{n}}$ near infinity is given by the
$n,n+1$ block $F$ of the matrix $q\wedge p$, which is in fact conserved
by the isospectral flow. The block $F$ depends only on the angular
momentum $K=q_{n}p_{n+1}-q_{n+1}p_{n}$ and is a regular matrix 
\[
F=\left(\begin{array}{cc}
0 & K\\
-K & 0\end{array}\right)
\]
 of rank $2$ if $K\not=0$. The isospectral flow given by \eref{eq:Lax}
conserves both matrix $A$ and the block $F$, therefore it is reasonable
to consider isospectral manifold of all matrix polynomials with this
data fixed 
\[
\M_{P}^{A,F}=\{L(\lambda)\in\M_{P};\quad L(\infty)=A,\ \mathrm{pr}_{E_{n}}\circ L'(\infty)|_{E_{n}}=F\}
\]
 where $\mathrm{pr}_{E_{n}}$is a projection to $E_{n},$ and the
values of $L(\lambda)$ and $L'(\lambda)$ at infinity are defined
by the limits $L(\infty):=\lim_{\lambda\to\infty}(\lambda^{-2}L(\lambda))$
and $L'(\infty):=\lim_{\lambda\to\infty}\rmd(\lambda^{-2}L(\lambda))/\rmd(\lambda^{-1})$.
On $\M_{P}^{A,F}$ acts the subgroup $\PG_{A,F}<\PG_{A}$ of matrices
that stabilize $F$ as well. We can write 
\[
\PG_{A,F}\simeq(\CC^{*})^{n-1}\times G_{F}
\]
where $G_{F}$ is the stabilizer subgroup of $F$ in $Gl(2,\CC)$.
The group $G_{F}\simeq\CC^{*}$ consists of matrices 
\[
r\left(\begin{array}{cc}
\cos\varphi & \sin\varphi\\
-\sin\varphi & \cos\varphi\end{array}\right),
\]
where $r\in(0,\infty)$ and $\varphi\in[0,2\pi)$. We will describe
the isospectral manifold $\M_{P}^{A,F}$ as an open subset of the
generalized Jacobian of singularized spectral curve.

\begin{thm}
\label{thm:Arnold-Jacobian}Let $f\in\CC^{n}$ be the value of $F_{EM}$
such that the normalized spectral curve $\C$ is smooth and $K\not=0$.
Let denote by $\TT_{A}=(\CC^{*})^{n-1}$ the subgroup of $\PG_{A}$
of diagonal matrices $G=[g_{i,j}]$ with $g_{n,n}=g_{n+1,n+1}$. 
\begin{enumerate}
\item The complex level set $F_{EM}^{-1}(f)$ is a covering of the isospectral
manifold $\M_{P}^{A,F}/\TT_{A}$. The fiber of the covering is the
same as the orbit of the group $(\ZZ_{2})^{n-1}$ generated by reflections
$q_{i}\mapsto-q_{i}$ for $1\le i\le n-1$.
\item The isospectral manifold $\M_{P}^{A,F}/\TT_{A}$ is isomorphic to
the open subset of the generalized Jacobian of the singular spectral
curve $\C_{m'}$ given as a singularization of smooth spectral curve
by modulus $m=(\infty,\rmi K)+(\infty,-\rmi K)$. 
\item \label{enu:flow-K}The flow of $K$ generates the fiber $\CC^{*}$
of the extension $\CC^{*}\to Jac(\C_{m})\to Jac(\C)$.
\item \label{enu:flow-H}The flow of $H$ on the generalized Jacobian $Jac(\C_{m})$
is linear.
\end{enumerate}
\end{thm}
\begin{proof}
If $K\not=0$ then $F$ is a regular matrix with 2 different eigenvalues
$\pm\rmi K$. By the result in \cite{Vivolo:EdMathSoc:2000} the isospectral
manifold $\M_{P}^{A,F}$ is isomorphic as an algebraic manifold to
the Zariski open subset $J(\C_{m})-\Theta$ of the generalized Jacobian
of the singularized spectral curve, given by spectral polynomial $P$
and modulus $m'=\lambda^{-1}(\infty)=\sum_{i=1}^{n-1}(\infty,a_{i})+(\infty,\rmi K)+(\infty,-\rmi K)$.
The same theorem asserts that the following diagram commutes
\[
\xymatrix{0\ar[r] & \PG_{A,F}\ar[r]\ar[d] & \M_{P}^{A,F}\ar[r]\ar[d]_{e_{m'}} & \M_{P}^{A,F}/\PG_{A,F}\ar[r]\ar[d]_{e} & 0\\
0\ar[r] & (\CC^{*})^{n}\ar[r] & Jac(\C_{m'})-\Theta\ar[r] & Jac(\C)-\Theta\ar[r] & 0}
\]
If we singularize $\C$ only by modulus $m=(\infty,\rmi K)+(\infty,-\rmi K)$,
we can insert $e_{m}:\M_{P}^{A,F}/(\CC^{*})^{n-1}\to Jac(\C_{m})-\Theta$
into above diagram 
\[
\xymatrix{ & G_{F}\ar[d]\\
\M_{P}^{A,F}\ar[r]\ar[d]_{e_{m'}} & \M_{P}^{A,F}/\TT_{A}\ar[r]\ar[d]_{e_{m}} & \M_{P}^{A,F}/\PG_{A,F}\ar[d]_{e}\\
Jac(\C_{m'})-\Theta\ar[r] & Jac(\C_{m})-\Theta\ar[r] & Jac(\C)-\Theta\\
 & \CC^{*}\ar[u]}
\]
We only have to prove that the fiber $F_{EM}^{-1}(f)$ consisting
of matrices of the form \eref{eq:lax_n+1} is a covering of the quotient
$\M_{P}^{A,F}/\TT_{A}$. By using the Lax matrix \eref{eq:lax_n+1}
we gave a parametrization 
\[
J^{A}:(q,p)\mapsto L(\lambda)=A\lambda^{2}+q\wedge p\ \lambda-q\otimes q
\]
of the quotient $\M_{P}^{A,F}/\TT_{A}$ by $(q,p)\in F_{EM}^{-1}(f)\subset(\tst)^{\CC}$.
First note that the map $J^{A}:(\tst^{n})^{\CC}\to\M_{P}^{A,F}$ is
an immersion. We will show that any orbit of $\TT_{A}$ intersects
the image of $J^{A}$ only in finite number of points. To explain
how $\PG_{A,F}$ acts on the Lax matrix \eref{eq:lax_n+1}, note
that an element $g\in\PG_{A,F}$ acts on a tensor product $x\otimes y$
of $x,y\in\CC^{n+1}$ 
\[
g:x\otimes y\mapsto(gx)\otimes((g^{-1})^{T}y)
\]
by multiplying the first factor with $g$ and the second with $(g^{-1})^{T}$.
The subgroup of $\PG_{A,F}$, for which the generic orbit lies in
the image of $J^{A}$, is given by orthogonal matrices 
\[
O(n,\CC)\cap\PG_{A,F}\simeq(\ZZ_{2})^{n-2}\times G_{F}\cap O(2,\CC).
\]
There are special points in the image of $J^{A}$ that have a large
isotropy group (take for example $q_{i}=\delta_{ij}$ and $p_{i}=\delta_{ik}$,
$k\not=j$, where the isotropy is $(\CC^{*})^{n-2}$). But the intersection
of any orbit with the image of $J^{A}$ coincides with the orbit of
$(\ZZ_{2})^{n-2}\times G_{F}\cap O(2,\CC)$. If we take the torus
$\TT_{A}<\PG_{A,F}$ consisting of diagonal matrices with $g_{n,n}=g_{n+1,n+1}$,
so that $\TT_{A}\cap G_{F}=\{Id\}$, the orbits of $\TT_{A}$ will
intersect image of $J^{A}$ only in the orbit of the finite subgroup
$(\ZZ_{2})^{n-2}$. We have proved that the level set of Lax matrices
$L(\lambda)=J^{A}(q,p)$ with fixed characteristic polynomial $P$
is an immersed submanifold in $\M_{P}^{A,F}$ that intersects the
orbits of torus $\TT_{A}\simeq(\CC^{*})^{n-1}$ in only finite number
of points and is thus a covering of the quotient $\M_{P}^{A,F}/\TT_{A}$.

The group $G_{F}$ is the complexification of the group of rotations
in $q_{n},q_{n+1}$ plane and is generated by the Hamiltonian vector
field of $K$. This proves the assertion \eref{enu:flow-K}.

To prove the assertion \eref{enu:flow-H}, note that the matrix polynomial
$M(\lambda)$ in Lax equation \eref{eq:Lax} is given as a polynomial
part $R(\lambda,L(\lambda))_{+}$ for a polynomial $R(z,w)=zw$. It
is well known that such isospectral flows are mapped by $e_{m}$ to
linear flows on the Jacobian $Jac(\C_{m})$ (see \cite{audin:CUP:1999}
for reference).
\end{proof}
Taking into account the real structure on $\C_{m}$, the Arnold-Liouville
tori can be described as a real part of the generalized Jacobian. 

\begin{thm}
\label{thm:real-ext}For $K\not=0$ and $\C$ smooth, the Arnold-Liouville
tori are $(\ZZ_{2})^{n-2}$ coverings of the real part of the generalized
Jacobian. The rotations generated by $K$ are precisely the rotations
of the fiber $S^{1}$ in the fibration $S^{1}\to Jac(\C_{m})^{\RR}\to Jac(\C)^{\RR}$,
which is the real part of the fibration $\CC^{*}\to Jac(\C_{m})\to Jac(\C)$.
\end{thm}
Above theorem gives us an algebraic way to describe symplectic quotient
$\tst^{n}//_{k}S^{1}$. Algebraically $\tst^{n}$ is a covering of
the relative generalized Jacobian, which is a disjoint union 
\[
\tilde{J}ac^{\RR}=\cup_{\C_{m}}Jac^{\RR}(\C_{m})
\]
over the space of curves $\C_{m}$ corresponding to the real values
of energy momentum map. The symplectic quotient of $\tilde{J}ac^{\RR}//_{k}S^{1}$
is then the relative Jacobian over the space of normalized spectral
curves with fixed value of $K$
\[
\cup_{\C_{m};\, K=k}Jac^{\RR}(\C).
\]

\begin{cor}
\label{cor:reduced-jac(C)}The complex level set of $(F_{1},\ldots,F_{n-1})$
of the symplectic quotient of the confluent Neumann system is a $(\ZZ_{2}^{})^{n-2}$
covering of the quotient $\M_{P}^{A,F}/\PG_{A,F}$. The manifold $\M_{P}^{A,F}/\PG_{A,F}$
is isomorphic to the open subset $Jac(\C)-\Theta$ of the Jacobian
of the normalized spectral curve. 
\end{cor}
This result agrees perfectly with the results obtained previously
for the Rosochatius system \cite{Hurtubise:CMP:1990}.

\begin{proof}
[Proof of the theorem \ref{thm:real-ext}]Note that the eigenvalues
of $F$ are $\pm\rmi K$. Note also that the value of $\mu$ at the
points $P_{\pm}$ equals to the eigenvalues of $F$, so $P_{\pm}=(\infty,\pm\rmi K)$.
On $\C$ there is a natural real structure $J$ induced by the conjugation
on $(\lambda,\mu)\in\CC^{2}$. The points $P_{\pm}$ that are glued
in the singular point form a conjugate pair $P_{\pm}=JP_{\mp}$. If
we follow the argument in \cite{audin:CMP:2002} we can find the real
structure of the fiber $\CC^{*}$ in the extension $\CC^{*}\to Jac(\C_{m})\to Jac(\C)$.
Note that $Pic(\C)$ is defined as the space of all divisors modulo
divisors of meromorphic functions on $\C$, whereas $Pic(\C_{m})$
is given by the divisors on $\C$ that avoid $P_{\pm}$ modulo meromorphic
functions on $\C_{m}$. So the fiber $\CC^{*}$ is given by meromorphic
functions on $\C$ modulo meromorphic functions on $\C_{m}$. Since
we obtained $\C_{m}$ by gluing two points $P_{\pm}$, a function
$f$ on $\C$ defines a function on $\C_{m}$ if $f(P_{+})=f(P_{-})$
or equivalently $f(P_{+})/f(P_{-})=1$. For a divisor of any function
$f$ on $\C$, the number $z=f(P_{+})/f(P_{-})\in\CC^{*}$ determines
its class in the Picard group $Pic(\C_{m})$. So if $P_{\pm}=JP_{\mp}$,
then the real structure on the fiber $\CC^{*}$ is given by the map
\[
z=\frac{f(P_{+})}{f(P_{-})}\to\frac{\overline{f(JP_{+})}}{\overline{f(JP_{-})}}=\frac{\overline{f(P_{-})}}{\overline{f(P_{+})}}=\frac{1}{\bar{z}}
\]
and the real part of $\CC^{*}$ is the unit circle $S^{1}$ given
by $z\bar{z}=1$. In contrast, when the singular points are real $P_{\pm}=JP_{\pm}$,
the real structure on $\CC^{*}$ is given by the conjugation and the
real part of $\CC^{*}$ is $\RR^{*}$. 
\end{proof}

\subsection{Bifurcation diagram}

We can use the normalized spectral curve $\C$ to describe the singular
locus of energy momentum map.

\begin{prop}
\label{pro:bif-diag}The vector $(f_{1},\ldots,f_{n-1},k)\in\RR^{n}$
is a regular value of the real momentum map 
\[
F_{EM}=(F_{1},\ldots,F_{n-1},K)
\]
if and only if the normalized spectral curve $\C$ is smooth. The
singular locus of the map $F_{EM}$ consist of 
\begin{itemize}
\item hyperplanes $F_{i}=0$
\item zero level set of the discriminant of $Q(x)$ from \eref{eq:qx}
\item the codimension 2 hyperplane defined by $K=0$ and $H=\frac{1}{2}a_{i}$
\end{itemize}
\end{prop}
\begin{proof}
The if part follows directly from theorem \ref{thm:Arnold-Jacobian}.
If $\C$ is smooth and $K\not=0$, then the level set of $F_{EM}$
is locally isomorphic to the real part of $Jac(\C_{m})-\Theta$, which
is of dimension $n$. Hence the rank of the differential of $F_{EM}$
is also $n$. If $K=0$ and $\C$ is smooth, the level set of $(F_{1},\ldots,F_{n-1})$
is locally isomorphic to the real part of the isospectral manifold
$\M_{P}^{A,F}/\PG_{A,F}$, which is in turn isomorphic to the real
part of $Jac(\C)-\Theta$. Since the dimension of $Jac(\C)$ is $n-1$,
the rank of the differential of $(F_{1},\ldots,F_{n-1})$ is also
$n-1$. The integrals $F_{i}$ are invariant to the rotations generated
by $K$ and therefore their Hamiltonian vector fields $X_{F_{1}},\ldots,X_{F_{n-1}}$
are independent from $X_{K}$. We are left to show that if $X_{K}=0$
the normalized curve $\C$ is not smooth. It is easy to see that the
case $X_{K}=0$ appears only if $X_{H}=0$ but then the rank of $X_{F_{1}},\ldots,X_{F_{n-1}}$
is not full and the curve $\C$ has to be singular. 

To prove the only if part let us consider case by case the components
of the singular locus. The curve $\C$ is singular if and only if
the polynomial $\prod(a_{i}-x)Q(x)$ has a double root. This can
happen in two cases
\begin{enumerate}
\item $a_{i}$ is a zero of $Q(x)$, this happens when $F_{i}=0$
\item $Q(x)$ has double zero, this happens if the discriminant of $Q$
is zero.
\end{enumerate}
The hyperplanes $F_{i}=0$ are singular, since for the points $(q,p)$
with $q_{i}=p_{i}=0$ the differential $\rmd F_{i}=0$. The proof
that the discriminant of $Q$ is singular can be found in \cite{Audin:ExM:1994}
and I will omit it here, because it is very specific and beyond the
scope of this article. 
\end{proof}
\begin{rem}
For values of $a_{n}>a_{j}$ for some $1\le j<n$ the bifurcation
diagram has a singular ``thread'' of focus-focus singularities
defined by values $K=0$ and $H=\frac{1}{2}a_{j}$. This would suggest
the presence of nontrivial monodromy. Indeed for two degrees of freedom,
the singular level set corresponding to the isolated singular value
is a union of two spheres with two pairs of points identified. By
the general result in \cite{Zung:CompM:2003} it follows that the
monodromy is nontrivial and equals 
\[
\left(\begin{array}{cc}
1 & 0\\
\varepsilon & 1\end{array}\right)
\]
with $\varepsilon=2$ being the number of spheres in the singular
level set. This can also be checked by direct calculation.
\end{rem}
The image of the energy momentum map $(K,2H)$ for confluent Neumann system
with two degrees of freedom is depicted in \fref{fig:bif}.
Regular values lie in the shadowed area while the solid curves contain
singular values. The points $(0,a_{1})$ and $(0,a_{2})$ are the
images of fixed points. Note that when passing from $a_{2}>a_{1}$
to $a_{1}>a_{2}$, the pair of lines becomes imaginary and only their
intersection - isolated focus-focus point $(0,a_{1})$ - remains real. 

\begin{figure}
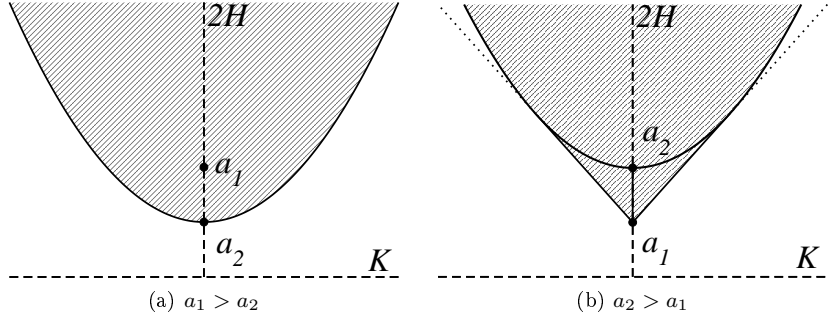

\begin{centering}
\subfloat[\(a_1>a_2\)]{\includegraphics[clip,width=0.4\textwidth]{figure1a}}
\quad
\subfloat[\(a_2>a_1\)]{\includegraphics[clip,width=0.4\textwidth]{figure1b}}
\caption{\label{fig:bif}Image of the energy-momentum map $(K,2H)$ for confluent Neumann system
with potential matrix $A=\mathrm{diag}(a_{1},a_{2},a_{2})$. Regular
values lie in the shadowed area while the solid curves contain singular
values. The points $(0,a_{1})$ and $(0,a_{2})$ are the images of
fixed points.}
\end{centering}
\end{figure}

\subsection{Note on the case $K=0$}

We have seen in the previous section that the case $K=0$ is significantly
different from the generic case $K\not=0$. The problem lies in the
following observation. The space of hyperelliptic curves that appear
in the description of Neumann system is parametrized by $K^{2}$ and
not by $K$. As a consequence, the relative generalized Jacobian is
degenerated for $K=0$. The phase space of complexified Neumann system
is therefore ``folded'' into relative Jacobian by the map $K\to K^{2}$.
The map between original phase space and its image in relative Jacobian
is singular at $K=0$. It is therefore illusory to expect that we
can describe the whole phase space including the fiber $K=0$ by algebro-geometric
methods. Different approach has to be considered that would study
the ``fold'' given by $K^{2}$ in more detail. This is to be covered
in our future work.

\section{Conclusions and discussion}

We have proved the algebraic integrability of the confluent Neumann
system by proving the theorem \ref{thm:Arnold-Jacobian}, which describes
Arnold-Liouville tori in terms of the generalized Jacobians of singular
spectral curve. We performed the reduction of the rotational symmetry
and established a firm relationship between symplectic reduction and
desingularization of the spectral curve (corollary \ref{cor:reduced-jac(C)}).
Most of our results very likely generalize to any Moser system arising
from the rank 2 perturbations of a fixed matrix with a double eigenvalue.
From our work and previous examples \cite{audin:CMP:2002,Zhivkov:Ens:1998}
it appears that there generally is a relation between the rotational
symmetry and singularities of spectral curves. We have exposed this
relationship explicitly in our case and have seen that the reduction
of $S^{1}$ symmetry reveals itself in the algebraic description as
a reduction from generalized Jacobian of the singular spectral curve
to the {}``ordinary'' Jacobian of the normalized spectral curve.
One can say that the desingularization of the spectral curve corresponds
to the symplectic quotient. Unfortunately, this relation is not a
general phenomenon as we can see when considering the case $K=0$.
One can speculate that the appearance of the global action of a compact
group is related to the presence of a {}``generic'' singularity
but there is no general proof yet. Note that the generic Neumann system
has no symmetries given by a compact group. 

The singularities of the spectral curve appeared in two different
roles in our study. The singularity that is a consequence of the confluency
is ``generic'' in that it appears uniformly for all values of
the energy-momentum map. The rotational symmetry shows as the extension
by complexified group of rotations $\CC^{*}$ defining the generalized
Jacobian that appears globally. The ``sporadic'' singularities,
which correspond to the singular values of the energy momentum map
(see proposition \ref{pro:bif-diag}) are strictly a local phenomenon.
In those cases the extension by $\CC^{*}$ and the resulting rotational
symmetry does not extend globally. Algebraically speaking both singularities
are the same, but the ``generic'' singularity appears globally
and thus give rise to a rotational symmetry. Sporadic singularities
on the other hand appear when the level sets of energy momentum map
are singular (orbits of lower dimension, heteroclinic and homoclinic
orbits). It would be interesting to describe the isospectral sets
of the singular spectral curves. Note that when we introduced generic
singularity we made sure that we used the subset of the singular isospectral
set, consisting of regular Lax matrices. In the study of sporadic
singularities, non regular part of the isospectral set should not
be avoided. It is our conjecture that the singular isospectral sets
that induce homoclinic or heteroclinic orbits should pose an obstruction
to the existence of global action of compact groups.

In a somewhat more ambitious and speculative vein, one could study
the relationship between symmetries of certain PDE's and generic singularities
in appropriate spectral curves of infinite genus. Maxwell-Bloch equation
for example can be viewed as a chain of confluent Neumann systems\cite{Saksida:JPA:2005,Saksida:Nonlin:2006},
whose symmetries indeed reflect in a symmetry of the whole Maxwell-Bloch
system \cite{Saksida:SIGMA:2006}. The description of Maxwell-Bloch
system with generalized Jacobians of singular spectral curves of infinite
genus should be to some extent analogous to our results.

\ack{}{I would like to thank Pavle Saksida for proposing and discussing
the subject and Michèle Audin for fruitful discussions and hospitality
while visiting \emph{Institut de Recherche Mathématique Avancée} at
Strasbourg where part of this work was done. I would also like to
acknowledge the financial support of the French government. }

\section*{References}{\bibliographystyle{unsrt}
\bibliography{martin}
}
\end{document}